\documentclass{birkjour}
 \newtheorem{thm}{Theorem}[section]
 \newtheorem{cor}[thm]{Corollary}
 \newtheorem{lem}[thm]{Lemma}
 \newtheorem{prop}[thm]{Proposition}
 \theoremstyle{definition}
 \newtheorem{defn}[thm]{Definition}
 \theoremstyle{remark}
 \newtheorem{rem}[thm]{Remark}
 
 \numberwithin{equation}{section}
 \newcommand{\eqa}{\begin{eqnarray}}
\newcommand{\eeqa}{\end{eqnarray}}
\newcommand{\beq}{\begin{equation}}
\newcommand{\eeq}{\end{equation}}
\newcommand{\nn}{\nonumber}
\newcommand{\pal}{\partial}

\begin{document}

%
%
%
%
%
%
%
%
%

\title[SFT and Schur polynomials]
 {Symplectic field theory of a disk, quantum integrable systems, and Schur polynomials}

\author[B.~Dubrovin]{Boris Dubrovin}

\address{%
SISSA\\
Via Bonomea, 265\\
I-34136 TRIESTE\\
ITALY\\
and N.N.Bogolyubov Laboratory of Geometrical Methods in Mathematical Physics,\\
Moscow State University, Moscow, RUSSIA}

\email{dubrovin@sissa.it}



\date{June 15, 2014}

\begin{abstract}
We consider commuting operators obtained by quantization of Hamiltonians of the Hopf (aka dispersionless KdV) hierarchy. Such operators naturally arise in the setting of Symplectic Field Theory (SFT). A complete set of common eigenvectors of these operators is given by Schur polynomials. We use this result for computing the SFT potential of a disk. 
\end{abstract}

\maketitle
\section{Introduction}

\subsection{Hopf integrable hierarchy and its quantization}

Hopf equation
\beq\label{hopf}
u_t+u\, u_x=0
\eeq
can be considered as the spatially one-dimensional analogue of the Euler equations of motion of ideal incompressible fluid. It
is perhaps the simplest example of an integrable nonlinear PDE. In the community of experts in integrable systems it is often called dispersionless Korteweg--de Vries (KdV) equation; sometimes it is also called (inviscid) Burgers equation. Eq. \eqref{hopf} is an infinite-dimensional analogue of a Hamiltonian dynamical system
$$
u_t+\pal_x \frac{\delta H}{\delta u(x)} =0, \quad H=\frac1{2\pi} \int_0^{2\pi} \frac{u^3(x)}6\, dx
$$
(to be more specific let us consider \eqref{hopf} as a dynamical system on the space of smooth $2\pi$-periodic functions). It possesses a complete family of pairwise commuting conservation laws
\beq\label{hopfh}
H_n=\frac1{2\pi}\int_0^{2\pi}\frac{u^{n+2}(x)}{(n+2)!}\, dx, \quad n\geq -1
\eeq
\begin{align*}
&
\{ H_n, H_m\}:= \frac1{2\pi} \int_0^{2\pi} \frac{\delta H_n}{\delta u(x)} \pal_x \left( \frac{\delta H_m}{\delta u(x)}\right)\, dx
\\
&
\\
&
=\frac1{2\pi}\int_0^{2\pi}\pal_x \left( \frac{u^{m+n+2}(x)}{(n+1)!\, m! \,(m+n+2)}\right)\, dx=0.
\end{align*}
The Hopf Hamiltonian coincides with the functional $H_1$; the Hamiltonian $H_0$ generates spatial translations; the very first Hamiltonian $H_{-1}$ is a Casimir of the degenerate Poisson bracket defined by the operator $\pal_x$.
A convenient generating function of the commuting Hamiltonians can be chosen in the form
\beq\label{hopfg}
{\mathcal H}(z) =\frac1{2\pi} \int_0^{2\pi}e^{z\, u(x)}dx=1+\sum_{n\geq -1} H_n z^{n+2}.
\eeq

In general the procedure of quantization of a Hamiltonian system on a symplecic manifold converts functions on the manifold into operators acting on a suitable space of states. It depends on the choice of canonical coordinates on the phase space.

In our case one can use the Fourier coefficients of the $2\pi$-periodic function
$$
u(x) =\sum_{n\in \mathbb Z}u_n e^{i\, n\, x}
$$
as coordinates on the phase space. Their Poisson brackets read
$$
\{ u_n, u_m\} =i \, n\, \delta_{m, -n} .
$$
Thus one can choose positive Fourier coefficients as canonical coordinates on the symplectic leaves $u_0={\rm const}$
$$
q_n=u_n, \quad n\geq 1.
$$
Negative Fourier coefficients will be denoted
$$
p_n =u_{-n}, \quad n\geq 1.
$$
Their Poisson brackets have the canonical form up to a constant factor
\beq\label{pb1}
\{ p_n, q_m\} =-i\, n\, \delta_{mn}, \quad \{ u_0, q_n\} =\{u_0, p_n\}=0.
\eeq

After the quantization one has to replace the variables $q_n$, $p_m$ by operators $\hat q_n$, $\hat p_m$ satisfying the commutation relations
\beq\label{pb2}
\left[\hat p_n, \hat q_m\right] =\hbar\, n\, \delta_{mn}.
\eeq
Here $\hbar$ is an independent parameter of the quantization procedure called Planck constant. 

The operators $\hat q_k$, $\hat p_k$ admit the following realization on the space of polynomials in the variables $q_1$, $q_2$, \dots
\beq\label{pq}
\hat q_k f(q)=q_k f(q), \quad \hat p_k f(q) = \hbar \, k\frac{\partial }{\partial q_k} f(q), \quad k=1, \, 2, \dots.
\eeq

We can associate certain operators acting on the same space of polynomials also with the Hamiltonians $H_n$ applying the procedure of normal ordering. That means that, in any monomial in the $p$- and $q$-variables all $p_n$-s must be written on the right and, then, replaced by operators $\hat p_n=\hbar \, n\, \frac{\pal}{\pal q_n}$. Denote 
\beq\label{h0}
\hat H_n^0= \frac1{2\pi} :\int_0^{2\pi} \frac{u^{n+2}(x)}{(n+2)!} dx: 
\eeq
the resulting \emph{normally ordered} quantum operators. For example,
\begin{align*}
&
\hat H_{-1}^0=u_0\\
&
\hat H_0^0=\frac1{2\pi} : \int_0^{2\pi} \frac{u^2}2dx: =\frac{u_0^2}2+\sum_{n\geq 1} \hat q_n \hat p_n 
\\
&
\hat H_1^0=\frac1{2\pi}:\int_0^{2\pi}  \frac{u^3}6  \, dx: =\frac{u_0^3}6+u_0 \sum\hat q_k \hat p_k+\frac12\sum_{i, \, j} \left( \hat q_i \hat q_j \hat p_{i+j} + \hat q_{i+j} \hat p_i \hat p_j\right)
\end{align*}
etc. In the semiclassical limit $\hbar \to 0$ the operators are replaced by their symbols clearly coinciding with the original Hamiltonians $H_n$.

However, these quantum operators do not commute any more; this can be seen already in the commutator
$$
\left[ \hat H_1^0, \hat H_2^0\right]=\frac{\hbar^2}8 \sum i\, j(i+j) \left(\hat q_{i+j} \hat p_i \hat p_j -\hat q_i \hat q_j \hat p_{i+j}\right)\neq 0.
$$
We arrive at the following quantization problem.
\begin{defn} Quantization of the integrable system \eqref{hopf} is a family of pairwise commuting operators $\hat H_n$, $n\geq -1$, acting on the Fock space of polynomials $\mathbb C[q_1, q_2, \dots]$ and, in the semiclassical limit $\hbar\to 0$, coinciding with $H_n$ given by  \eqref{hopfh} .
\end{defn}

Such quantization problem\footnote{Quantization of the KdV equation equipped with the \emph{second} Poisson bracket was first considered in a series of papers by V.~Bazhanov, S.~Lukyanov and A.~Zamolodchikov \cite{blz}. More recently the problem of quantization of the intermediate long wave equation was addressed in \cite{bz}.} was addressed in \cite{pogreb} where a recursion procedure for constructing commuting quantum Hamiltonians has been developed. Independently a closed formula for the commuting quantum Hamiltonians was found by Y.~Eliashberg \cite{private}, \cite{e}, see also \cite{rossi}, in the setting of Symplectic Field Theory (see below). This formula plays the central role in the present paper. 

\begin{prop} The operators $\hat H_n$, $n\geq -1$ defined by their generating function
\begin{equation}\label{generalh}
\hat {\mathcal H}(z, u_0,\hbar) =\frac1{2\pi\, s\left( \hbar^{1/2} z\right)} : \int_0^{2\pi}e^{z\, s\left( i\hbar^{1/2} z\partial_x\right)u(x)}dx:=1+\sum_{n\geq -1}\hat  H_n z^{n+2}
\end{equation}
where
$$
u(x)=u_0+\sum_{k\geq1} \left( \hat q_k e^{ikx} +\hat p_k e^{-ikx}\right)
$$
and
$$
s(t) =\frac{\sinh \frac{t}2}{\frac{t}2}
$$
commute pairwise
$$
\left[ \hat H_n, \hat H_m\right]=0, \quad n, \, m\geq -1.
$$
\end{prop}

\begin{rem}The commuting operators $\hat H_n$ were recently rediscovered in essentially equivalent form in \cite{a2010}, \cite{kp}. 
\end{rem}

\subsection{Schur polynomials and eigenvectors of the quantum Hamiltonians}

Using the Taylor expansion of the function $s(t)$
$$
s(t)=1+\sum_{n\geq 1} \frac{t^{2n}}{2^{2n} (2n+1)!}
$$
one can spell out the formula \eqref{generalh} as follows
$$
\hat {\mathcal H}(z, u_0,\hbar)=\frac1{2\pi\, s\left( \hbar^{1/2} z\right)} : \int_0^{2\pi} e^{z\, u-\frac{\hbar\, z^3}{24} u''+\frac{\hbar^2 z^5}{1920} u^{(4)}-\dots}dx:
$$
so, in the semiclassical limit $\hbar \to 0$, one returns to the Hamiltonians \eqref{hopfh} but, starting from $\hat H_2$ there are nontrivial corrections
\begin{align*}
&
\hat H_{-1}=u_0\\
&
\hat H_0=\frac1{2\pi} : \int_0^{2\pi} \frac{u^2}2dx: -\frac{\hbar}{24}
\\
&
\hat H_1=\frac1{2\pi}:\int_0^{2\pi}  \frac{u^3}6 \, dx: -\frac{\hbar}{24} u_0
\\
&
\hat H_2=\frac1{2\pi} :\int_0^{2\pi} \left[ \frac{u^4}{24}+\frac{\hbar}{24} \left( u\, u''-\frac12 u^2 \right)\right]\, dx: +\frac{7\hbar^2}{5760}.
\end{align*}
Observe that, in the representation \eqref{pq} the nontrivial part of $\hat H_0$ coincides with the degree operator, up to an additive constant,
$$
\hat H_0|_{u_0=0}=\hbar \sum_n n\, q_n \frac{\pal}{\pal q_n}-\frac{\hbar}{24}
$$
and $\hat H_1$ with the cut-and-join operator
$$
\hat H_1|_{u_0=0}=\frac12 \sum_{i,\, j} \left( \hbar\, (i+j) q_i q_j \frac{\pal}{\pal q_{i+j}} +\hbar^2 i\, j \, q_{i+j} \frac{\pal^2}{\pal q_i \pal q_j}\right).
$$

It was observed in \cite{kl} that Schur polynomials are eigenvectors\footnote{It can also be deduced from the description \cite{amos} in terms of Jack polynomials of eigenstates of the Calogero--Sutherland operator at the limit $\beta\to 1$. I thank P. Wiegman for this reference.} of the cut-and-join operator. It turns out that the same is true for all commuting Hamiltonians \eqref{generalh}. In order to formulate the precise statement\footnote{A closely related result was recently announced in \cite{a2010}.} let us remind basic definitions of the theory of Schur polynomials.

Introduce polynomials $h_k(q)$ by
\begin{equation}\label{elem}
\sum_{k\geq 0} h_k(q) z^k =e^{\sum_{k\geq 1} q_k \frac{z^k}{k}}, \quad h_k=0\quad \mbox{for}\quad k<0.
\end{equation}
For any Young tableau $\lambda=(\lambda_1\geq \lambda_2\geq \dots \geq \lambda_n>0)$ (the number $n$ for the tableau $\lambda$ will be denoted $l(\lambda)$) define a polynomial
$$
s_\lambda(q)=\det \left( h_{\lambda_i-i+j}(q)\right)_{1\leq i, \, j\leq l(\lambda)}.
$$
They coincide with the standard Schur polynomials $s_\lambda(x)$, up to a normalization of the independent variables $x_k =\frac{q_k}{k}$. 

\begin{thm} \label{thm1} For any partition $\lambda$ the polynomial $ s_\lambda\left(\frac{q}{\hbar^{1/2}}\right)$ is an eigenvector of the operator \eqref{generalh}
\eqa\label{gener}
&&
\hat {\mathcal H}(z, u_0,\hbar)  s_\lambda\left(\frac{q}{\hbar^{1/2}}\right)={ E}(z,\lambda, u_0, \hbar) s_\lambda\left(\frac{q}{\hbar^{1/2}}\right)
\nn\\
&&
\\
&&
{ E}(z,\lambda, u_0, \hbar)=\hbar^{1/2} z\, e^{z\, u_0}\sum_{i= 1}^\infty e^{z\, \hbar^{1/2} \left( \lambda_i -i+\frac12\right)}, \quad {\rm Re}\, z>0.
\nn
\eeqa
\end{thm}

The sum in the right-hand side makes sense for any infinite sequence of nonnegative integers $\lambda_1\geq \lambda_2\geq\dots$ such that $\lambda_i=0$ for $i\gg 0$. Observe that the eigenvalue \eqref{gener} essentially coincides with the Chern character of the partition $\lambda$ introduced in \cite{mn}
$$
{ E}(z,\lambda, u_0, \hbar=1)=z\, {\rm ch}_\lambda(u_0, z).
$$

All Schur polynomials $s_\lambda(q)$ of degree $n=|\lambda|$ form a basis in the space $V_n\subset \mathbb C[q]$ of all graded homogeneous polynomials of degree $n$ in the variables $q_1$, \dots, $q_n$, where the degree is assigned according to the following rule
$$
\deg q_i=i, \quad i=1, \, 2, \dots.
$$
Thus, the theorem \ref{thm1} provides us with a complete family of eigenstates of the commuting quantum Hamiltonians.

\subsection{Applications to Symplectic Field Theory}

Symplectic Field Theory (SFT) is an approach to constructing topological invariants of contact manifolds and symplectic cobordisms between them by enumeration of pseudoholomorphic maps of open Riemann surfaces with certain boundary conditions. It also provides one with a surgery tool in computing Gromov--Witten (GW) invariants of closed symplectic manifolds by gluing them from simple blocks. 

With a compact contact\footnote{This assumption can be slightly relaxed, see \cite{e} for details.} manifold $V$ the SFT associates a bosonic Fock space ${\mathcal F}(V)$ generated by closed Reeb orbits on $V$, assuming them to be all nodegenerate. Choosing a closed odd-dimensional differential form $\theta$ on $V$ one obtains an infinite sequence of pairwise commuting quantum Hamiltonians $\hat H_k^V$ acting on the Fock space. For the particular case $V=S^1$ this gives \cite{private, e} the quantum Hamiltonians \eqref{generalh}, see more details in Section \ref{sft} below.

A symplectic cobordism $W$ between two contact manifolds $V_+$ and $V_-$, 
$$
\pal W=V_+ \cup (-V_-)
$$
defines an operator
\beq\label{canon}
Z^W_{\rm SFT}: {\mathcal F}(V_+)\to {\mathcal F}(V_-)
\eeq
called the total SFT potential of $W$. For a closed symplectic manifold $W$, $\pal W=\emptyset$, the SFT potential coincides with the total Gromov--Witten (GW) descendent potential of $W$.

In particular, to a symplectic manifold $W$ with the boundary $V$ the total SFT potential of $W$ defines a state $\Psi_{\rm SFT}^W$ in the Fock space ${\mathcal F}(V)$ (or in its dual, depending on the orientation of $V=\pm \pal W$). A suitable choice of an even-dimensional closed differential form $\phi=\theta\wedge d\rho$ on $W$ makes the SFT potential $\Psi_{\rm SFT}^W$ depend on an infinite number of coupling parameters $t_k^\phi$. The dependence on these parameters is determined from a system of pairwise commuting Schr\"odinger\footnote{Due to our normalization of canonical coordinates and Hamiltonians the $\sqrt{-1}$ does not appear in the equations.} equations
$$
\hbar \frac{\pal}{\pal t_k^\phi} \Psi_{\rm SFT}^W=\hat H_k^V \Psi_{\rm SFT}^W, \quad k=0, \,1,\dots.
$$ 
Our goal is to compute this SFT potential for the simplest example $V=S^1$, $W$=disk.

Denote
\beq\label{taylor}
{ E}(z,\lambda, u_0, \hbar)=1+\sum_{n\geq -1}{ E}_n(\lambda, u_0, \hbar) z^{n+2}
\eeq
the Taylor expansion of the eigenvalues \eqref{gener}, see the explicit formula \eqref{eigen} below.

\begin{thm} \label{thm2} The total SFT potential of a disk $\subset \mathbb C$ is given by the following formula
\begin{equation}\label{cap} 
\Psi_{\rm SFT}^{\rm disk} \left(u_0, {\bf t}, {\bf p}; \hbar\right)=\sum_\lambda \hbar^{-\frac{|\lambda|}2}\frac{\dim \lambda}{|\lambda|!} e^{\frac1{\hbar} \sum_{k\geq 0} t_k { E}_k(u_0, \hbar,\lambda)}  s_\lambda\left(\frac{p}{\hbar^{1/2}}\right).
\end{equation} 
\end{thm}

\medskip

Collecting in \eqref{cap} partitions of a given number $n=|\lambda|$ one obtains a part of the SFT potential corresponding to maps of Riemann surfaces with cylindrical ends of a given degree $n$. For low degrees $|\lambda|\leq 3$ keeping the variables $t_k$ for $k\leq 3$ one obtains (for simplicity we set $u_0=0$; recall that this parameter is responsible for adding punctures on the Riemann surface with no cohomology classes inserted)
\begin{eqnarray}
&&
\Psi_{\rm SFT}^{\rm disk}=e^{-\frac{ t_0}{24}+\frac{7 \hbar\, t_2}{5760}}\left[1+\frac1{\hbar}e^{t_0 +\frac{\hbar\, t_2}{24}} p_1
\right.
\nonumber\\
&&
+\frac1{2\hbar^2} \left( e^{2 t_0 + \hbar^{\frac12} t_1 +\frac{7 \hbar\, t_2}{12} +\frac{5}{24}\hbar^{\frac32} t_3} (p_1^2 + \hbar^{\frac12} p_2)
+e^{2 t_0 - \hbar^{\frac12} t_1 +\frac{7 \hbar\, t_2}{12} -\frac{5}{24}\hbar^{\frac32} t_3} (p_1^2 - \hbar^{\frac12} p_2)\right)
\nonumber\\
&&
+\frac1{36 \hbar^3} \left( e^{3 t_0 + 3\hbar^{\frac12} t_1 +\frac{21}8 \hbar\, t_2 +\frac{13}8 \hbar^{\frac32} t_3} (p_1^3 + 3\hbar^{\frac12} p_1 p_2 + 2\hbar \, p_3)+4 e^{3 t_0 +\frac98 \hbar\, t_2}(p_1^3-\hbar\, p_3)
\right.
\nonumber\\
&&\left.\left.
+e^{3 t_0 - 3\hbar^{\frac12} t_1 +\frac{21}8 \hbar\, t_2 -\frac{13}8 \hbar^{\frac32} t_3} (p_1^3 - 3\hbar^{\frac12} p_1 p_2 + 2\hbar \, p_3)
\right)
\right]+\dots
\nonumber
\end{eqnarray}
It is not difficult to see that the resulting expansion contains only integer powers of Planck constant due to the following well known properties (see, e.g., \cite{mac})
$$
s_{\lambda'}(p) =(-1)^{|\lambda|} s_\lambda(-p)
$$
$$
\dim\lambda'=\dim\lambda.
$$
Here $\lambda'$ is the transposed Young tableau.

\begin{rem} Setting in \eqref{cap} $u_0=0$, $t_1=\beta$, $t_k=0$ for $k\neq 1$ one obtains the I.P.~Goulden and D.M.~Jackson generating function for Hurwitz numbers \cite{gj} (see also \cite{vak}) in the representation of M.~Kazarian and S.~Lando \cite{kl}.
\end{rem}

\begin{cor} \label{cor1.6} The tau-function $\tau({\bf p}, \epsilon)=\Psi_{\rm SFT}^{\rm disk}({\bf t}, {\bf p}, u_0,\epsilon^2)$ for arbitrary values of the parameters $u_0$ and ${\bf t}=(t_0, t_1, \dots)$ satisfies equations of the KP hierarchy
\beq\label{kp}
\sum_{j\geq 0} h_j(-2{\bf y}) h_{j+1}(\epsilon\, \tilde D) e^{\epsilon\sum_{k\geq 1} y_k D_k}\tau\cdot\tau=0 \quad {\rm for ~all}~{\bf y}=(y_1, y_2, \dots).
\eeq
Here
$$
D=(D_1, D_2, \dots)
$$
are the Hirota bilinear operators,
$$
D_k\,\tau\cdot\tau=\left[\frac{\pal}{\pal q} \tau(p_k+q)(\tau(p_k-q)\right]_{q=0},
$$
$\tilde D=(D_1, 2D_2, 3 D_3, \dots)$.
\end{cor}

Recall that, expanding eqs. \eqref{kp} with respect to the indeterminates $y_1$, $y_2$, etc. one obtains an infinite family of bilinear equations for tau-function of the KP hierarchy. E.g., the first two equations read
\begin{align*}
&
(12 D_2^2 -12 D_1 D_3 +\hbar\, D_1^4)\,\tau\cdot\tau=0
\\
&
(6 D_2 D_3 -6 D_1 D_4 + \hbar\, D_1^3 D_2)\,\tau\cdot\tau=0
\end{align*}
(recall: here and below $\epsilon^2 =\hbar$).
The second logarithmic derivative
$$
u=\epsilon^2 \pal_x^2 \log\tau
$$
as function of $x=p_1$, $y=p_2$, $t=p_3$ satisfies the Kadomtsev--Petviashvili equation
\beq\label{kp1}
u_{xt}=u_{yy} +\left( u\, u_x +\frac{\hbar}{12}u_{xxx}\right)_x.
\eeq 
\medskip

Following the general scheme of \cite{egh} one can use the SFT potential of the disk for computation of the stationary sector (in the terminology of \cite{op}) of Gromov--Witten theory of projective line ${\bf P}^1$. Namely,

\begin{cor} \label{cor1.7} Denote $\overline{M}_{g,n}^\bullet({\bf P}^1, d)$ the space\footnote{Here we follow notations of \cite{op}.} of degree $d$ stable maps to ${\bf P}^1$  of not necessarily connected algebraic curves of the total genus $g$ with $n$ punctures. Let 
$$
\chi: \overline{M}_{g,n}^\bullet({\bf P}^1, d)\to \mathbb Z
$$ 
be a locally constant function taking value $\chi=2g-2k$ on the subset of stable maps of $k$-component curves. Then
\begin{eqnarray}\label{final}
&&
Z_{{\bf P}^1}^\bullet({\bf t}, z, \hbar):=
\nn\\
&&
\sum_{g\geq 0} \sum_{d\geq 0} z^d \sum_{m,\, n} \frac{u_0^m}{m!}  \sum_{k_1, \dots, k_n}\frac{t_{k_1}\dots t_{k_n}}{n!} \int\limits_{\left[\overline{M}_{g,n+m}^\bullet({\bf P}^1, d)\right]^{\rm virt}}  \hbar^{\frac{\chi}2} {\rm ev}_1^*(\omega)\dots {\rm ev}_n^*(\omega) \psi_1^{k_1}\dots \psi_n^{k_n} 
\nonumber\\
&&
=\sum_\lambda \left( \frac{z}{\hbar}\right)^{|\lambda|}\left( \frac{\dim\lambda}{|\lambda|!}\right)^2 e^{\frac1{\hbar} \sum_{k\geq 0} t_k { E}_k (\lambda, u_0, \hbar)}.
\end{eqnarray}
\end{cor}

In this formula $\omega\in H^2({\bf P}^1)$ is the area form normalized by $\int\limits_{{\bf P}^1}\omega=1$.

Specializing \eqref{final} at $u_0=0$, $\hbar=1$ and collecting the terms of degree $d$ one arrives at eq. (0.25) of the A. Okounkov and R. Pandharipande paper \cite{op} playing an important role in their analysis of GW/Hurwitz correspondence.

\medskip

The paper is organised as follows. In Section \ref{sec1} we recall basics about the so-called boson-fermion correspondence. We apply this technique for proving Theorem \ref{thm1}. In Section \ref{sft} we very briefly remind main definitions of Symplectic Field Theory and prove  Theorem \ref{thm2} as well as Corollaries \ref{cor1.6} and \ref{cor1.7}.

\section{Boson-fermion correspondence and proof of Theorem \ref{thm1}}\label{sec1}

For simplicity let us do the calculations setting $\hbar= 1$. The original normalization will be restored at the end of the proof.

We will use the so-called boson-fermion correspondence. Let us briefly remind the main step of this construction; the reader can find more details in \cite{jmd}, \cite{az}. By definition the algebra of free bosons has generators\footnote{Hats over the bosonic operators will be omitted in this section.} $q_n$, $p_n$, $n\geq 1$, satisfying commutation relations
\begin{equation}\label{bose}
[p_n, q_m]=n\, \delta_{n,m}, \quad [p_n, p_m]=[q_n, q_m]=0.
\end{equation}
It admits a natural representation on the bosonic Fock space 
\begin{align*}
&
\mathbb C[q_1, q_2, \dots]\,{\hat{}}=\bigoplus_{n=0}^\infty V_n
\\
&
V_n =\{ f\in \mathbb C[q_1, \dots, q_n]\, | \, \deg f=n\}, \quad \mbox{assuming}\quad\deg q_i=i.
\end{align*}
where the generator $q_n$ acts by multiplication by $q_n$ and $p_n$ by derivation
$$
p_n=n\,\frac{\partial}{\partial q_n}.
$$
The algebra of free fermions has two infinite sequences of generators $\psi_k$ and $\psi_k^*$ labeled by half-integers. The defining relations are given in terms of anticommutators
$$
\{ a, b\}=ab+ba.
$$
They read
\begin{equation}\label{fermi}
\{ \psi_i, \psi_j^*\}=\delta_{ij}, \quad \{\psi_i, \psi_j\}=\{\psi^*_i, \psi_j^*\}=0.
\end{equation}
Fermionic Fock space ${\mathcal F}$ is spanned by vectors obtained by action of finite products of the generators on the vacuum vector $|\,0>$. It is assumed that
$$
\psi_{-k}|\,0>=0, \quad \psi_k^*|\,0>=0 \quad \mbox{for}\quad k>0.
$$
The generators $\psi_{-k}$ and $\psi^*_k$ for $k>0$ are called annihilation operators; the generators $\psi_k$ and $\psi_{-k}^*$ are creation operators. For any product of the fermionic generators the normal order is defined by moving all annihilation operators on the right. For example, the normal ordering  of the product $\psi_i \psi_i^*$
is
$$
:\psi_i \psi_i^*:=\left\{ \begin{array}{rc}\psi_i \psi_i^*, & i>0\\
-\psi_i^* \psi_i, & i<0\end{array}\right. .
$$

Vectors of charge zero in the fermionic Fock space can be written as linear combinations of vectors of the form
$$
\psi_{i_1}\dots \psi_{i_n} \psi_{j_1}^* \dots \psi_{j_n}^*|\,0>
$$
for an arbitrary integer $n\geq 0$. The subspace in the Fock space spanned by these vectors will be denoted ${\mathcal F}_0$. The dual space ${\mathcal F}^*$ is obtained by a similar action on the bra-vector $<0\,|$ assuming that
$$
0=<0\,| \psi_k, \quad 0= <0\,|\psi_{-k}^*\quad \mbox{for}\quad k>0.
$$
The pairing ${\mathcal F}^*\times {\mathcal F}\to\mathbb C$ is completely defined by the anticommutation relations \eqref{fermi} along with the normalization
$$
<0\,|\, 0>=1.
$$

Boson-fermion correspondence is an isomorphism of linear spaces 
$$
\Phi: {\mathcal F}_0\to \mathbb C[q_1, q_2, \dots]\,\hat{}
$$
given by the formula
\begin{eqnarray}\label{bofe}
&&
\Phi(\xi\,|\,0>)=<0\,| e^{K(q)} \xi \, |\,0>
\\
&&
K(q) =\sum_{n=1}^\infty \frac{q_n}{n} \sum_j : \psi_j \psi_{j+n}^*:
\nonumber
\end{eqnarray}
for any state $\xi\, |\,0>\in{\mathcal F}_0$.

The following formulae hold true
\begin{eqnarray}\label{psitoqp}
&&
e^{K(q)} \psi_i e^{-K(q)} =\psi_i+h_1(q) \psi_{i-1} +h_2(q) \psi_{i-2}+\dots
\nonumber\\
&&
\\
&&
e^{K(q)} \psi_i^* e^{-K(q)} =\psi_i^*+h_1(-q) \psi_{i+1}^* +h_2(-q) \psi_{i+2}^*+\dots
\nonumber
\end{eqnarray}
Here $h_k(q)$ are elementary Schur polynomials defined in \eqref{elem}.

Due to the isomorphism \eqref{bofe} any operator on ${\mathcal F}_0$ becomes an operator on the bosonic Fock space $\mathbb C[q_1, q_2, \dots]\,\hat{}\,$ and vice versa. For example,
the pull-back of the operators $q_n$ and $p_n$ acting on the bosonic Fock space is given by the following infinite sums
\begin{eqnarray}
&&
\Phi^{-1} q_n\, \Phi =\sum_{j\in\mathbb Z+\frac12} : \psi_j \psi_{j-n}^*:
\nonumber\\
&&
\Phi^{-1} p_n \,\Phi =\sum_{j\in\mathbb Z+\frac12} : \psi_j \psi_{j+n}^*:
\nonumber
\end{eqnarray}

The following statement was proved by P. Rossi in \cite{rossi}.

\begin{prop} The fermionic representation of the Hamiltonian \eqref{generalh} reads
\begin{equation}\label{rossi}
\Phi^{-1} \hat{\mathcal H}(z, u_0, \hbar=1)\, \Phi = e^{z\, u_0} \left( z\sum_{k\in\mathbb Z+\frac12} e^{k\, z} : \psi_k \psi_k^*: +\frac1{s(z)}\right).
\end{equation}
\end{prop}

Following \cite{az} construct a basis of the fermionic Fock space ${\mathcal F}_0$ labeled by partitions $\lambda=(\lambda_1 \geq \lambda_2 \geq \dots \geq \lambda_n>0)$ of arbitrary length $n\geq 0$. It is convenient to use Frobenius notations for partitions
\beq\label{frob}
\lambda=(\alpha_1, \dots, \alpha_d \, |\, \beta_1, \dots, \beta_d)
\eeq
where $d=d(\lambda)$ is the length of the diagonal of the Young tableau $\lambda$, the numbers $\alpha_i$, $\beta_i$ are defined as follows
$$
\alpha_i=\lambda_i-i, \quad \beta_i =\lambda_i'-i, \quad i=1, \dots, d(\lambda).
$$
Here $\lambda'$ is the transposed Young tableau.
Define the state $|\,\lambda> \in {\mathcal F}_0$ by
\begin{equation}\label{lambda}
|\,\lambda> = \psi_{\alpha_1+\frac12} \psi_{\alpha_2+\frac12}\dots  \psi_{\alpha_{d(\lambda)}+\frac12}\psi_{-\beta_{d(\lambda)} -\frac12}^* \dots \psi_{-\beta_2-\frac12}^*\psi_{-\beta_1-\frac12}^*    |\,0>.
\end{equation}
For empty partition put $|\,\emptyset >=|\,0>$.

The following statement readily follows from Proposition 2.10 of \cite{az}.

\begin{prop} For any partition $\lambda$ one has
\begin{equation}\label{lambdabasis}
\Phi (|\,\lambda>)= (-1)^{b(\lambda)} s_\lambda(q)
\end{equation}
where
$$
b(\lambda)=\sum_{i=1}^{d(\lambda)} (\beta_i+1).
$$
\end{prop}

Denote
\begin{equation}\label{rossi0}
{\mathcal O}(z) =  \sum_{k\in\mathbb Z+\frac12} e^{k\, z} : \psi_k \psi_k^*:
\end{equation}
the nontrivial part of the operator \eqref{rossi}.

\begin{lem} The states $|\,\lambda>$ are eigenvectors of the Hamiltonian \eqref{rossi0}
\begin{equation}\label{rossi1}
{\mathcal O}(z) |\,\lambda>=  \sum_{i=1}^{d(\lambda)} \left[ e^{z\left( \alpha_i+\frac12\right)} - e^{-z\left( \beta_i+\frac12\right)}\right] \, |\,\lambda>.
\end{equation}
Here the partition $\lambda$ is represented in the Frobenius form \eqref{frob}.
\end{lem}

\begin{proof} One has
\begin{align*}
&
{\mathcal O}(z) |\,\lambda>=\sum_{k>0 }e^{k\, z} \psi_k \psi_k^* \psi_{\alpha_1+\frac12} \psi_{\alpha_2+\frac12}\dots  \psi_{\alpha_{d(\lambda)}+\frac12}\psi_{-\beta_{d(\lambda)} -\frac12}^* \dots \psi_{-\beta_2-\frac12}^*\psi_{-\beta_1-\frac12}^*    |\,0>
\\
&
-\sum_{k>0} e^{-k\, z} \psi_{-k}^* \psi_{-k}\psi_{\alpha_1+\frac12} \psi_{\alpha_2+\frac12}\dots  \psi_{\alpha_{d(\lambda)}+\frac12}\psi_{-\beta_{d(\lambda)} -\frac12}^* \dots \psi_{-\beta_2-\frac12}^*\psi_{-\beta_1-\frac12}^*    |\,0>.
\end{align*}
Since $\psi_k^*$ annihilates the vacuum for $k>0$, in the first sum the only nonzero terms are those for which $k=\alpha_i+\frac12$ for some $1\leq i\leq d(\lambda)$. So
\begin{align*}
&
\sum_{k>0 }e^{k\, z} \psi_k \psi_k^* \psi_{\alpha_1+\frac12} \psi_{\alpha_2+\frac12}\dots  \psi_{\alpha_{d(\lambda)}+\frac12}\psi_{-\beta_{d(\lambda)} -\frac12}^* \dots \psi_{-\beta_2-\frac12}^*\psi_{-\beta_1-\frac12}^*    |\,0>
\\
&=\sum_{i=1}^{d(\lambda)} e^{z\, \left(\alpha_i+\frac12\right)} |\lambda>.
\end{align*}
In a similar way in the second sum the nonzero terms come from $k=\beta_i+\frac12$, $1\leq i \leq d(\lambda)$. This gives the second half of eq. \eqref{rossi1}. \end{proof}

It is easy to see that
\beq\label{dva}
\sum_{i=1}^{d(\lambda)} \left[ e^{z\left( \alpha_i+\frac12\right)} - e^{-z\left( \beta_i+\frac12\right)}\right] =\sum_{i=1}^{l(\lambda)} \left[ e^{z\left( \lambda_i-i+\frac12\right)} -e^{z\left(-i+\frac12\right)}\right].
\eeq
So, to complete the proof of the theorem it remains to reinsert the Planck constant. This can be done by rescaling
$$
q_k \mapsto \hbar^{\frac{k-1}2} q_k, \quad z\mapsto \hbar^{\frac12} z.
$$
The theorem is proved.

\begin{cor} For any partition $\lambda=(\lambda_1\geq \lambda_2\geq \dots \geq \lambda_{l(\lambda)}>0)=(\alpha_1, \dots, \alpha_{d(\lambda)} \, |\, \beta_1, \dots, \beta_{d(\lambda)})$ the polynomials $ s_\lambda\left(\frac{q}{\hbar^{1/2}}\right)$ are common eigenvectors of the commuting operators $\hat H_k$
\beq\label{hk}
\hat H_k  s_\lambda\left(\frac{q}{\hbar^{1/2}}\right)={ E}_k (\lambda, u_0, \hbar) s_\lambda\left(\frac{q}{\hbar^{1/2}}\right) , \quad k\geq -1
\eeq
with the eigenvalues given by one of the two equivalent expressions 
\begin{eqnarray}\label{eigen}
&&
{ E}_k (\lambda, u_0, \hbar)=c_k(u_0, \hbar) 
\\
&&+\hbar^{\frac12} \sum_{i=1}^{l(\lambda)} \frac{ \left[ u_0+\hbar^{\frac12}\left( \lambda_i -i+\frac12\right)\right]^{k+1} -\left[ u_0+\hbar^{\frac12}\left(  -i+\frac12\right)\right]^{k+1} }{(k+1)!}
\nn\\
&&
=c_k(u_0, \hbar)+\hbar^{\frac12}\sum_{i=1}^{d(\lambda)} \frac{\left[ u_0+\hbar^{\frac12} \left( \alpha_i+\frac12\right)\right]^{k+1} -\left[ u_0-\hbar^{\frac12} \left( \beta_i+\frac12\right)\right]^{k+1}}{(k+1)!}
\nonumber
\end{eqnarray}
where
$$
c_k(u_0, \hbar)=-\frac1{(k+2)!} \sum_{j=0}^{k+2} \left(\begin{array}{c}k+2\\ j\end{array}\right) \left( 1-2^{1-j}\right) B_j \hbar^{\frac{j}2}u_0^{k-j+2},
$$
$B_j$ are Bernoulli numbers.
\end{cor}

\begin{proof} With the help of an obvious formula
$$
\sum_{i\geq 1} e^{z\, \left( -i+\frac12\right)} =\frac1{z\, s(z)}, \quad {\rm Re}\, z>0
$$
the expressions \eqref{gener} can be rewritten in the form
\begin{align*}
&
{ E}(z,\lambda, u_0, \hbar) =e^{z\, u_0} \left[ \frac1{s\left( \hbar^{\frac12} z\right)} +\hbar^{\frac12} z \sum_{i=1}^{l(\lambda)} \left( e^{z\,\hbar^{\frac12} \left( \lambda_i-i+\frac12\right) } -e^{z\, \hbar^{\frac12}\left( -i+\frac12\right)}\right)\right]
\\
&
=e^{z\, u_0} \left[ \frac1{s\left( \hbar^{\frac12} z\right)} +\hbar^{\frac12} z \sum_{i=1}^{d(\lambda)} \left( e^{z\,\hbar^{\frac12} \left( \alpha_i+\frac12\right) } -e^{-z\, \hbar^{\frac12}\left( \beta_i+\frac12\right)}\right)\right].
\end{align*}
Expanding the right-hand side in $z$ and using the Taylor expansion 
$$
\frac1{s(t)}=\sum_{n=0}^\infty \left(2^{1-n} -1\right)\frac{B_n}{n!} t^n
$$
one arrives at the needed formula for the eigenvalues.
\end{proof}

\begin{rem} \label{rem2.5} An alternative derivation \cite{op} of the spectrum of the operator \eqref{rossi0} uses a realization of the fermionic algebra by operators acting on the Sato infinite-dimensional Grassmannian. Let us briefly outline this construction that will be useful below. Introduce an infinite-dimensional space with a basis $e_k$,
$$
{\mathcal V}=\bigoplus_{k\in\mathbb Z+\frac12} \mathbb C e_k.
$$
Denote
\eqa\label{grass}
&&
\Lambda^{\frac{\infty}2} {\mathcal V}={\rm span} \left( e_{i_1}\wedge e_{i_2}\wedge\dots ~ | ~ i_1>i_2>\dots, \quad i_k=-k~ \mbox{for large}~ k\right),
\nn\\
&&
e_j\wedge e_i = - e_i \wedge e_j.
\eeqa
Action of the fermionic algebra is defined by adding or erasing a factor on the left of a semi-infinite wedge product
\beq\label{grass1}
\psi_k =e_k\wedge\,, \quad \psi_k^* =\frac{\pal}{\pal e_k}, \quad k\in\mathbb Z+\frac12.
\eeq
The vacuum vector reads
$$
|\,0>=e_{-1/2}\wedge e_{-3/2}\wedge\dots.
$$
For any partition $\lambda=(\lambda_1\geq \lambda_2\geq \dots \geq \lambda_{l}>0)$ consider
\beq\label{grass2}
v_\lambda= e_{\lambda_1-\frac12}\wedge e_{\lambda_2-\frac32}\wedge \dots \wedge e_{\lambda_l-l+\frac12}\wedge e_{-l-\frac12}\wedge\dots\in \Lambda^{\frac{\infty}2} {\mathcal V}.
\eeq
The vectors of the form \eqref{grass2} span the fermionic phase space ${\mathcal F}_0$.
According to \cite{op} it is an eigenvector of the operator \eqref{rossi0} with the eigenvalue
$$
\sum_{i=1}^{l} \left[ e^{z\left( \lambda_i-i+\frac12\right)} -e^{z\left(-i+\frac12\right)}\right].
$$
\end{rem}

\section{SFT potential of a disk}\label{sft}

Let us very briefly outline the main constructions of SFT; more details can be found in \cite{egh}, \cite{beh}, \cite{e}; see also \cite{rossi}. 

Let $(V,\alpha)$ be a contact manifold with a contact form $\alpha$. Denote $R$ the Reeb vector field on $V$. Let ${\mathcal P}(V)$ be the set of periodic orbits of $R$. For simplicity assume that all periodic orbits are non-degenerate; in this case ${\mathcal P}(V)$ is a discrete set. The first playing character is the \emph{SFT Hamiltonian} ${\bf H}^V$. The construction uses intersection theory on the moduli spaces of pseudoholomorphic maps
$$
f: (C_g, x_1, \dots, x_n;  y_1^-, \dots, y_{r_-}^-; y_1^+, \dots, y_{r_+}^+) \to V\times \mathbb R
$$
of Riemann surfaces of genus $g$ with $n$ marked points $x_1$, \dots, $x_n$, $r_-$ negative punctures $y_1^-$, \dots, $y_{r_-}^-\in C_g$ and $r_+$ positive punctures $y_1^+$, \dots, $y_{r_+}^+\in C_g$. One has to choose an almost complex structure $J$ on the cylinder $V\times \mathbb R$ compatible with the symplectic form $d\left( e^t\alpha\right)$ where $t\in\mathbb R$. It must be invariant with respect to $t$-translations and satisfy certain additional restrictions. The boundary conditions for the pseudoholomorphic maps near the positive/negative punctures depend on a choice of suitably oriented Reeb orbits $\gamma_1^+$, \dots, $\gamma_{r_+}^+$ and $\gamma_1^-$, \dots, $\gamma_{r_-}^-$ in ${\mathcal P}(V)$. Namely, the map $f$ must be asymptotically cylindrical near the negative punctures $y_1^-$, \dots, $y_{r_-}^-$ to the orbits $\gamma_1^-$, \dots, $\gamma_{r_-}^-$ at $t\to -\infty$ and to $\gamma_1^+$, \dots, $\gamma_{r_+}^+$ near the positive punctures $y_1^+$, \dots, $y_{r_+}^+$ at $t\to +\infty$. Denote 
\beq\label{moduli}
M_{g, n} (V; \mbox{\boldmath$\gamma$}^-, \mbox{\boldmath$\gamma$}^+), \quad \mbox{\boldmath$\gamma$}^+=(\gamma_1^-, \dots, \gamma_{r_-}^-),   \quad \mbox{\boldmath$\gamma$}^-=(\gamma_1^+, \dots, \gamma_{r_+}^+)
\eeq
the moduli space of such maps. The group $\mathbb R$ acts on the moduli space by translations. The quotient can be compactified by adding stable holomorphic buildings \cite{beh}. Denote $\overline{M_{g, n} (V; \mbox{\boldmath$\gamma$}^-, \mbox{\boldmath$\gamma$}^+)/\mathbb R}$ the compactified space. Choosing a system of closed differential forms $\theta_1$, \dots, $\theta_N$ on $V$ define the so-called \emph{SFT Hamiltonian} ${\bf H}^V={\bf H}^V({\bf t}, {\bf q}, {\bf p}, \hbar)$ as a generating function of the intersection numbers on the moduli spaces as follows\footnote{Our definition differs by a factor $\hbar$ from the one of \cite{egh}, \cite{e}. This factor will be restored below, see the Schr\"odinger equation \eqref{schroe0}.}
\eqa\label{sft1}
&&
{\bf H}^V=\sum_{g\geq 0} \hbar^g \sum_{n, r_-, r_+} \sum_{\substack{0\leq i_1, \dots, i_n\leq N\\
0\leq k_1, \dots, k_n}}\sum_{\substack{\gamma_{b_1}^-, \dots, \gamma_{b_{r_-}}^-\\
\gamma_{c_1}^+, \dots, \gamma_{c_{r_+}}^+}}\frac{t^{i_1}_{k_1}\dots t^{i_n}_{k_n}}{n!} \frac{q_{b_1}\dots q_{b_{r_-}}}{r_-!}\frac{p_{c_1}\dots p_{c_{r_+}}}{r_+!}
\nn\\
&&
\times \int\limits_{\overline{M_{g, n} (V; \,\mbox{\boldmath$\gamma$}^-, \mbox{\boldmath$\gamma$}^+)/\mathbb R}}{\rm ev}_1^* (\theta_{i_1})\psi_1^{k_1}\wedge\dots \wedge {\rm ev}_n^*(\theta_{i_n})\psi_n^{k_n}
\eeqa
Here ${\rm ev}_i$ are the evaluation maps
$$
{\rm ev}_i: \overline{M_{g, n} (V; \mbox{\boldmath$\gamma$}^-, \mbox{\boldmath$\gamma$}^+)/\mathbb R} \to V, \quad f \mapsto f(x_i), \quad i=1, \dots, n,
$$
$\psi_i=c_1\left( {\mathcal L}_i\right)$ are the Chern classes of the tautological line bundles over the moduli spaces.
All the 
independent variables in the formula are $\mathbb Z/2$-graded. In particular, the parity of the variables $t^i_k$, $i=1, \dots, N$, $k=0, \, 1, \dots$ is determined by parity of the chosen differential forms $\theta_1$, \dots, $\theta_N$. The $q$- and $p$-variables labeled by Reeb orbits belong to a Weyl (super)-algebra with the commutation relations
$$
\left[ p_\gamma, q_{\gamma'}\right]_{\pm} =\hbar\, {\rm mult}(\gamma) \delta_{\gamma\, \gamma'}.
$$
As the moduli spaces $\overline{M_{g, n} (V; \mbox{\boldmath$\gamma$}^-, \mbox{\boldmath$\gamma$}^+)/\mathbb R}$ are odd-dimensional, the generating function is an odd element of the algebra. 
Expanding ${\bf H}^V$ with respect to \emph{odd} variables $t^i_k$ and taking the cohomology class one obtains therefore infinite families of commuting Hamiltonians
$$
\hat H^V_{i,k} ({\bf q}, {\bf p}, \hbar):=\frac{\pal}{\pal t^i_k} \left[{\bf H}^V\right]_{{\bf t}=0}
$$
in the Weyl algebra.

In the simplest case $V=S^1$ the above construction yields the commuting quantum Hamiltonians  \eqref{generalh}. In this case  the set of Reeb orbits is just the set of integers. The moduli spaces \eqref{moduli} are nontrivial only for positive multiplicities. Thus one has independent variables $q_n$, $p_n$, $n=1, \, 2, \dots$ satisfying the commutation relation \eqref{pb2}. Choose $\theta_0=1$, $\theta_1= d\varphi$ as two differential forms on $S^1$ (here $\varphi$ is the angular coordinate on the circle).   The SFT Hamiltonian ${\bf H}^{S^1}$ is an odd element of the Weyl algebra with generators $p_n$, $q_n$, $n\geq 1$ depending on even variables $t^0_k$ and odd variables $t^1_k$, $k\geq 0$.  Taking derivatives with respect to odd variables produces the commuting quantum Hamiltonians according to the following procedure
 \cite{e}.

\begin{prop} Restriction of the derivatives of  ${\bf H}^{S^1}={\bf H}^{S^1}({\bf t}, q, p, \hbar)$ onto the locus $t^0_0=u_0$, $t^0_k=0$ for $k>0$, $t^1_k=0$ for all $k$ yields the commuting Hamiltonians \eqref{generalh}
\beq\label{hamk}
\frac{\pal}{\pal t^1_n}{\bf H}^{S^1}|_{t^0_0=u_0, \quad t^0_k=0~{\rm for}~k>0, ~ t^1_k=0}=\hat H_n(u_0,\hbar), \quad n\geq 0.
\eeq
\end{prop}

More generally, with a symplectic cobordism $W$ between two contact manifolds $V_+$ and $V_-$,
$$
\pal W=V_+ \cup (-V_-)
$$
one can associate a family of moduli spaces $M_{g,n}(W, \mbox{\boldmath$\gamma$}^-, \mbox{\boldmath$\gamma$}^+)$ labeled by collections of Reeb orbits $\mbox{\boldmath$\gamma$}^-=(\gamma_1^-, \dots, \gamma_{r_-}^-)$ on $V_-$ and $\mbox{\boldmath$\gamma$}^+=(\gamma_1^+, \dots, \gamma_{r_+}^+)$ on $V_+$. To this end one has to consider pseudoholomorphic maps with cylindrical ends of Riemann surfaces of genus $g$ with $r_-$ negative and $r_+$ positive punctures and also with $n$ marked points $x_1$, \dots, $x_n$ to the symplectic manifold $W\cup (V_-\times \left(-\infty, 0]\right)\cup \left(V_+\times [0,+\infty)\right)$ equipped with an appropriate almost complex structure. Evaluation maps are defined
$$
{\rm ev}_i: M_{g,n}(W, \mbox{\boldmath$\gamma$}^-, \mbox{\boldmath$\gamma$}^+)\to W, \quad i=1, \dots, n
$$
by taking images of the marked points. In this way one arrives at the SFT potential
\beq\label{sftpot}
Z_{\rm SFT}^W ({\bf t},{\bf q}, {\bf p}, \hbar)=e^{\frac1{\hbar} \sum_{g\geq 0} \hbar^g{\mathcal F}_g^{\rm SFT} \left( {\bf t}, {\bf q}, {\bf p}\right)}.
\eeq
Here the genus $g$ part ${\mathcal F}_g^{\rm SFT} \left( {\bf t}, {\bf q}, {\bf p}\right)$ is the generating function of numbers obtained by
integrating over the compactified moduli spaces $\overline{M_{g,n}(W, \mbox{\boldmath$\gamma$}^-, \mbox{\boldmath$\gamma$}^+)}$ the pull-backs ${\rm ev}_i^*(\phi)$ of differential forms $\phi$ on $W$ along with $\psi$-classes. The independent variables ${\bf q}=(q_{\gamma^-})$ and ${\bf p}=(p_{\gamma^+})$ are labeled by Reeb orbits on $V_-$ and $V_+$ respectively.

If $W$ is obtained by glueing together two symplectic cobordisms
$$
W=W_1 \cup W_2, \quad \pal W_1=V\cup (-V_-), \quad \pal W_2=V_+ \cup (-V)
$$
between contact manifold $V_-$ and $V$ and $V$ and $V_+$ respectively then the SFT potential of $W$ can be obtained by applying the following glueing rules
\begin{align}\label{glue}
&
Z^W({\bf q}_-, {\bf p}_+) =Z^{W_1}\left({\bf q}_-, \overrightarrow{\bf p}\right) Z^{W_2} \left( {\bf q}, {\bf p}_+\right)|_{{\bf q}=0}
\\
&
\quad\quad\quad\quad\quad\,\,\,=Z^{W_1}\left({\bf q}_-, {\bf p}\right) Z^{W_2} \left( \overleftarrow{\bf q}, {\bf p}_+\right)|_{{\bf p}=0}.
\nn
\end{align}
The variables ${\bf q}=(q_\gamma)$ and ${\bf p}=(p_\gamma)$ are labeled by Reeb orbits on $V$. The left-acting operator $\overrightarrow{p}_\gamma$ and right-acting operator $\overleftarrow{q}_\gamma$ are defined by
$$
\overrightarrow{p}_\gamma=\hbar\, {\rm mult}(\gamma)\, \frac{\pal}{\pal q_\gamma}, \quad \overleftarrow{q}_\gamma=\hbar\, {\rm mult}(\gamma)\, \frac{\pal}{\pal p_\gamma}.
$$

In the particular case of a symplectic manifold $W$ with a boundary $V=V_-$ the total SFT potential depends only on the canonical coordinates $q_\gamma$ labeled by Reeb orbits on $V$. It can be considered as a state in the bosonic Fock space generated by Reeb orbits on $V$. This state will be denoted by
$$
\Psi_{\rm SFT}^W({\bf t}, {\bf q}, \hbar):=Z_{\rm SFT}^W({\bf t}, {\bf q}, \hbar).
$$
In a similar way, for a symplectic manifold $W$ with a boundary $V=V_+$ one obtains a state in the dual Fock space
$$
\Psi_{\rm SFT}^W({\bf t}, {\bf p}, \hbar):=Z_{\rm SFT}^W({\bf t}, {\bf p}, \hbar).
$$
Finally, for a closed symplectic manifold $W$ the total SFT potential coincides with the GW total descendent potential
$$
Z_{\rm SFT}^W({\bf t}, \hbar)=Z_{\rm GW}^W({\bf t}, \hbar).
$$

To insert pullbacks of cohomology classes of the form $\phi=\theta\wedge d\rho$,  $\theta\in \Omega^{\rm odd}(V_+)$, $d\rho\in \Omega^1_{\rm comp}(\mathbb R)$ along with their descendents one can use the SFT Hamiltonians \eqref{sft1}:
\beq\label{schroe}
Z_{\rm SFT}^W ({\bf t}, {\bf q}_-,{\bf p}_+, \hbar) =\left( Z^W_{\rm SFT}({\bf q}_-, {\bf p}, \hbar)\, e^{\frac1{\hbar} \sum_{k\geq 0} t_k^\phi \hat H_k^{V_+} (\overleftarrow{\bf q}, {\bf p}_+, \hbar)}\right)_{{\bf p}=0}
\eeq
where the quantum Hamiltonians $H_k^{V_+}$ are defined in a way similar to \eqref{hamk},
$$
\hat H_k^{V_+} ({\bf q}, {\bf p}, \hbar)=\frac{\pal}{\pal t^\theta_k} {\bf H}^{V_+}({\bf t}, {\bf q}, {\bf p}, \hbar)|_{{\bf t}=0}, \quad k\geq 0.
$$

In the calculations of this section the following two formulae from the theory of Schur polynomials will be useful (see, e.g., the Macdonald book \cite{mac}).

\begin{prop} The following formulae hold true
\begin{equation}\label{form1}
q_1^{n} = \sum_{|\lambda|=n} {\rm dim}\, \lambda \cdot s_\lambda(q_1, \dots, q_n)
\end{equation}
\begin{equation}\label{form2}
s_\lambda(q_1, \dots, q_n) =\frac{{\rm dim}\, \lambda}{n!} q_1^n +\dots, \quad n=|\lambda|
\end{equation}
where periods stand for monomials of lower degree in $q_1$.
\end{prop} 

In these formulae ${\rm dim}\, \lambda$ is the dimension of the irreducible representation of the symmetric group $S_n$ associated with the Young tableau $\lambda$. Recall that this dimension can be computed using the hook length formula
$$
{\rm dim}\, \lambda =\frac{|\lambda|!}{\prod_{(i,j)\in\lambda} h(i,j)}
$$
({\it ibid}.) 

We are now ready to prove Theorem \ref{thm2}.

\begin{proof} According to the prescription \eqref{schroe} one has to solve a system of pairwise commuting non-stationary Schr\"odinger equations for the wave function $\Psi=\Psi_{\rm SFT}^{\rm disk}$
\beq\label{schroe0}
\hbar \frac{\partial}{\partial t_k} \Psi =\Psi \, \overleftarrow{H_k}, \quad k=0, \, 1, \, 2, \dots
\eeq
with the plane wave initial data
$$
\Psi|_{{\bf t}=0} = e^{\frac{p_1}{\hbar}}.
$$
Using eq. \eqref{form1} one decomposes the initial condition in a linear combination of common eigenvectors of the commuting Hamiltonians 
$$
e^{\frac{p_1}{\hbar}}=\sum_\lambda \hbar^{-\frac{|\lambda|}2}\frac{\dim \lambda}{|\lambda|!}  s_\lambda\left(\frac{p}{\hbar^{1/2}}\right).
$$
With the help of such a decomposition, using \eqref{eigen}, one arrives at \eqref{cap}. \end{proof}

We will now prove Corollary \ref{cor1.6} .

\begin{proof} According to \cite{jmd} tau-functions of the KP hierarchy correspond to simple multivectors
$$
\tau(q_1, q_2, \dots)=\Phi\left(\varphi_{-\frac12}\wedge \varphi_{-\frac32}\wedge \dots\right)
,\quad 
\varphi_\nu  \in {\mathcal V}, \quad \nu \in\mathbb Z+\frac12, \quad \nu<0
$$
(for notations see Remark \ref{rem2.5} above). Choosing
\begin{align*}
&
\varphi_\nu =\sum_{i\geq 0} A_{\nu,i} e_{\nu+i}
\\
&
A_{\nu, i}=\frac1{i!} \exp\left\{ \hbar^{\frac12}\sum_{k\geq 0} t_k \frac{[u_0+\hbar^{\frac12}(\nu+i)]^{k+1}-[u_0+\hbar^{\frac12} \nu]^{k+1}}{(k+1)!}
\right\}
\end{align*}
(cf. \cite{kaz}) one obtains the expression \eqref{cap}.
\end{proof}

{\it Proof} of Corollary \ref{cor1.7}. According to the above procedures to compute the Gromov--Witten potential of ${\bf P}^1$ within the stationary sector one has to apply the SFT potential of the disk to the potential of the cap without insertions, i.e.,
$$
Z_{{\bf P}^1}^\bullet({\bf t}, z, \hbar)= \Psi_{\rm SFT}^{\rm disk} \left(u_0, {\bf t}, \overrightarrow{\bf p}; \hbar\right)e^{\frac{z\,q_1}{\hbar}}|_{q=0}.
$$
Here the independent parameter $z$ is added for convenience in order to separate terms of various degrees. Due to \eqref{form2} one has
$$
 s_\lambda (\overrightarrow{\bf p}) q_1^n|_{q=0}=\left\{\begin{array}{cl} \hbar^n\dim\lambda, & |\lambda|=n\\
 0, & \mbox{otherwise.}\end{array}\right. 
 $$
 This completes the proof.

\subsection*{Acknowledgment}
I am indebted to Y. Eliashberg for introducing me into Symplectic Field Theory. I also wish to thank G.~Bonelli, A.~Tanzini, and P.~Wiegmann for stimulating discussions. This work is partially supported by the PRIN 2010-11 Grant ``Geometric and analytic theory of Hamiltonian
systems in finite and infinite dimensions'' of Italian Ministry of Universities
and Researches and by 
Russian Federation Government Grant No. 2010-220-01-077.

\end{document}